\documentclass[conference]{IEEEtran}

\usepackage{amsthm}
\usepackage{amsmath}
\usepackage{amssymb}
\usepackage{amsmath}
\usepackage{amssymb}
\usepackage{epsfig}
\usepackage{epsf}
\usepackage{subfigure}
\usepackage{graphicx}
\usepackage{url}
\usepackage{cite}
\usepackage{color}

\usepackage{verbatim}
\usepackage{environ}

\usepackage[numbers,sort&compress]{natbib}

\usepackage{enumitem,kantlipsum}

\def\BibTeX{{\rm B\kern-.05em{\sc i\kern-.025em b}\kern-.08em
    T\kern-.1667em\lower.7ex\hbox{E}\kern-.125emX}}

\newtheorem{theorem}{Theorem}
\newtheorem{definition}{Definition}

\newtheorem{lemma}{Lemma}

\newtheorem{lemmarep}{Lemma}

\newcommand{\1}{{\bf 1}}

\newcommand{\ep}{\epsilon}

\newcommand{\A}{{\mathcal A}}
\newcommand{\B}{{\mathcal B}}
\newcommand{\C}{{\mathcal C}}

\newcommand{\E}{{\mathcal E}}

\newcommand{\X}{{\mathcal X}}
\newcommand{\vX}{{\vec{X}}}
\newcommand{\Y}{{\mathcal Y}}
\newcommand{\vY}{{\vec{Y}}}

\newcommand{\x}{{\bf x}}

\newcommand{\aln}[1]{\begin{align*}#1\end{align*}}

\newcommand{\al}[1]{\begin{align}#1\end{align}}

\newcounter{numcount}
\newcommand{\leqnum}{\stackrel{(\roman{numcount})}{\leq\;}\stepcounter{numcount}}

\newcommand{\cnt}{$(\roman{numcount})$ \stepcounter{numcount}}
\newcommand{\rescnt}{\setcounter{numcount}{1}}

\newif\iflong
\longfalse

\newif\ifdraft
\drafttrue

\begin{document}

\title{
Communicating over the Torn-Paper Channel
}

\author{
\IEEEauthorblockN{Ilan Shomorony}
\IEEEauthorblockA{
ECE Department \\
	University of Illinois at Urbana-Champaign\\
	ilans@illinois.edu
	\vspace{-7mm}} 
	\and
\IEEEauthorblockN{Alireza Vahid}
\IEEEauthorblockA{Department of Electrical Engineering\\
University of Colorado Denver\\
alireza.vahid$@$ucdenver.edu \vspace{-7mm}}
	}
\maketitle

\begin{abstract}
We consider the problem of communicating over a channel that randomly ``tears'' the message block into small pieces of different sizes and shuffles them.
For the binary torn-paper channel with block length $n$ and pieces of length ${\rm Geometric}(p_n)$, we characterize the capacity as $C = e^{-\alpha}$, where $\alpha = \lim_{n\to\infty} p_n \log n$.
Our results show that the case of ${\rm Geometric}(p_n)$-length fragments and the case of deterministic length-$(1/p_n)$ fragments are qualitatively different and, surprisingly, the capacity of the former is larger.
Intuitively, this is due to the fact that, in the random fragments case, large fragments are sometimes observed, which boosts the capacity.
\end{abstract}


\section{Introduction}
\label{Section:Introduction}

Consider the problem of transmitting a message by
writing it on a piece of paper, which will be torn into small pieces of random sizes and randomly shuffled.
This coding problem is illustrated in Figure~\ref{fig:channel}. 
We refer to it as the \emph{torn-paper coding}, in allusion to the classic dirty-paper coding problem \cite{dirtypaper}.

This problem is mainly motivated by macromolecule-based (and in particular DNA-based) data storage, which has recently received significant attention due to several proof-of-concept DNA storage systems  \cite{church_next-generation_2012,goldman_towards_2013,grass_robust_2015,bornholt_dna-based_2016,erlich_dna_2016,organick_scaling_2017}. 
In these systems, data is written onto synthesized DNA molecules, which are then stored in solution.
During synthesis and storage, molecules in solution are subject to random breaks and, due to the unordered nature of macromolecule-based storage, the resulting pieces are shuffled \cite{heckel_characterization_2018}.
Furthermore, the data is read via high-throughput sequencing technologies, which is typically preceded by physical fragmentation of the DNA with techniques like \emph{sonication} \cite{pomraning2012library}.
In addition, the torn-paper channel is related to the DNA shotgun sequencing channel, studied in \cite{MotahariDNA,BBT,gabrys2018unique}, but in the context of variable-length reads, which are obtained in nanopore sequencing technologies \cite{laver2015assessing,mao2018models}.

\begin{figure}[b]
\vspace{-6mm}
\centering
\includegraphics[width=0.70\linewidth]{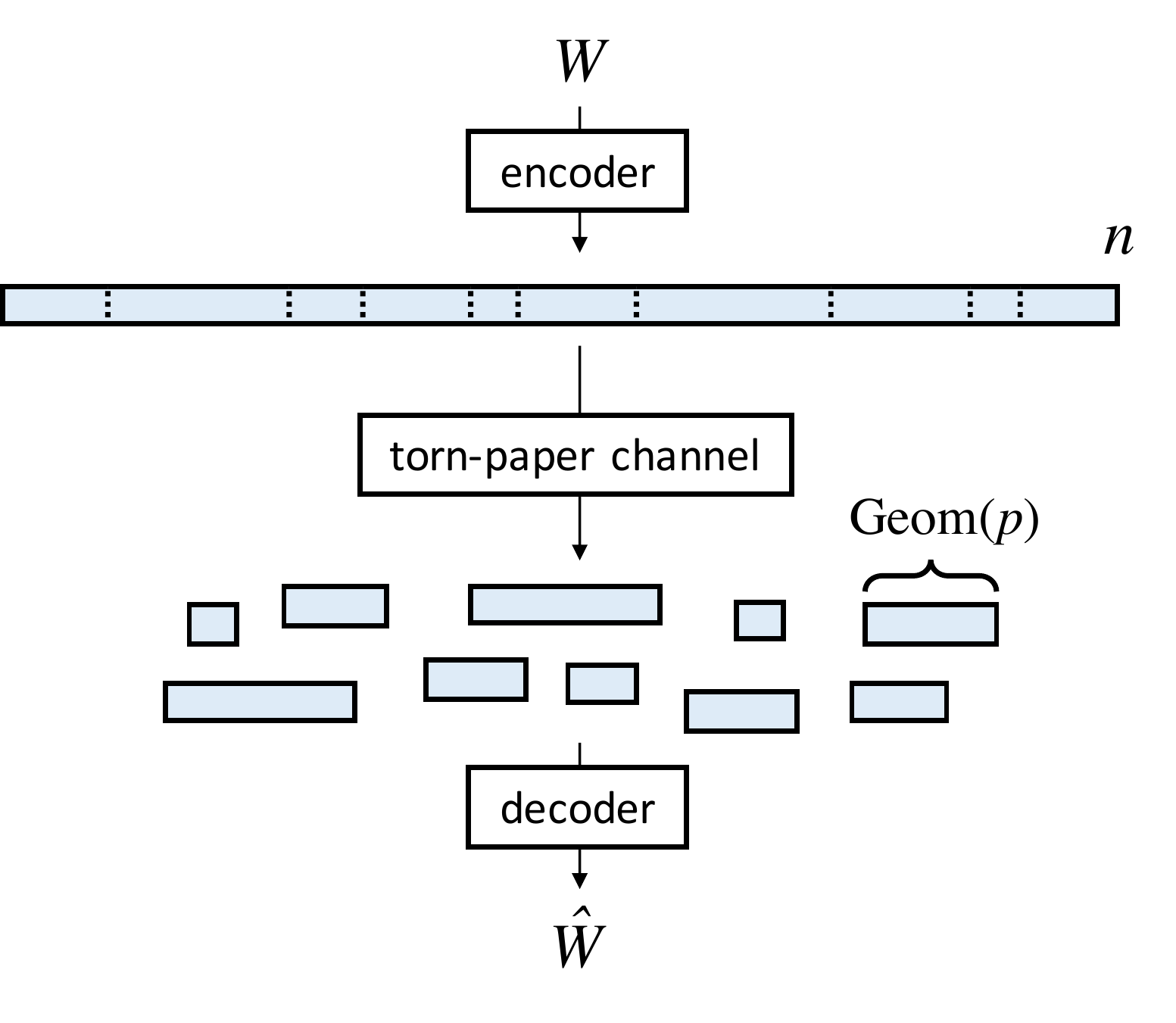}
\vspace{-4mm}
\caption{The torn-paper channel.
\label{fig:channel}
}
\end{figure}

We consider the scenario where the channel input is a length-$n$ binary string, which is then torn into pieces of lengths $N_1,N_2,...$, each of which has a ${\rm Geometric}(p_n)$ distribution.
The channel output is the unordered set of these pieces.
As we will see, even this noise-free version of the torn-paper coding problem is non-trivial.

To obtain some intuition, notice that $E[N_i] = 1/p_n$, and hence it is reasonable to compare our problem to the case where the tearing points are evenly separated, and $N_i = 1/p_n$ for $i=1,2,...,n p_n$ with probability $1$.
In this case, the channel becomes a \emph{shuffling channel}, similar to the one considered in \cite{noisyshuffling}, but with no noise. 
Coding for the case of deterministic fragments of length $N_i = 1/p_n$ is easy: since the tearing points are known, we can prefix each fragment with a unique identifier, which allows the decoder to correctly order the $n p_n$ fragments.
From the results in \cite{noisyshuffling}, such an index-based coding scheme is capacity-optimal, and any achievable rate in this case must satisfy, for large $n$,
\al{
R < (1 - p_n \log n)^+.
\label{eq:capshuffling}
}
If we let $\alpha = \lim_{n\to \infty} p_n \log n$, the capacity for this case becomes $(1-\alpha)^+$.

It is not clear a priori whether the capacity of the torn-paper channel should be higher or lower than $(1-\alpha)^+$.
The fact that the tearing points are not known to the encoder makes it challenging to place a unique identifier in each fragment, suggesting that the torn-paper channel is ``harder'' and should have a lower capacity.
The main result of this paper contradicts this intuition and shows that the capacity of the torn-paper channel with ${\rm Geometric}(p_n)$-length fragments is higher than $(1-\alpha)^+$.
More precisely, we show that the capacity of the torn-paper channel is 
$C = e^{-\alpha}$.
Intuitively, this boost in capacity comes from the tail of the geometric distribution, which guarantees that a fraction of the fragments will be significantly larger than the mean $E[N_i] = 1/p_n$.
This allows the capacity to be positive even for $\alpha \geq 1$, in which case the capacity of the deterministic-tearing case in (\ref{eq:capshuffling}) becomes $0$.


\section{Problem Setting}
\label{Section:Problem}

We consider the problem of coding for the torn-paper  channel, illustrated in Figure~\ref{fig:channel}.
The transmitter encodes a message $W \in \{1,...,2^{nR}\}$ into a length-$n$ binary codeword $X^n \in \mathbb{F}_2^n$.
The channel output is a set of binary strings
\begin{align}
\mathcal{Y} = \left\{ \vY_1, \vY_2, \ldots, \vY_K \right\}.
\end{align}
The process by which $\mathcal{Y}$ is obtained is described next. 
\begin{enumerate}[wide]
\item The channel tears the input sequence into segments of ${\rm Geometric}(p_n)$-length 
for a \emph{tearing probability} $p_n$. 
More specifically, let $N_1,N_2,...$ be i.i.d.~Geometric$(p_n)$ random variables. 
Let $K$ be the smallest index such that
$
\sum_{i=1}^{K}{N_{i}}  \geq n.
$
Notice that $K$ is also a random variable. 

The channel tears $X^n$ into segments $\vX_1,...,\vX_K$, where 
\aln{
\vX_i = \left[X_{1+\sum_{j=1}^{i-1}{N_j}},...,X_{\sum_{j=1}^{i}{N_j}} \right],
}
for $i=1,...,K-1$ and
\aln{
\vX_K = \left[X_{1+\sum_{j=1}^{K-1}{N_j}},...,X_{n} \right].
}
We note that this process is equivalent to independently tearing the message in between consecutive bits with probability $p_n$.
More precisely, let $T_2,T_3,...,T_n$ be binary indicators of whether there is a cut between $X_{i-1}$ and $X_i$. 
Then, letting $T_i$s be i.i.d.~${\rm Bernoulli}(p_n)$ random variables results in independent fragments of length ${\rm Geometric}(p_n)$.
Also, $K = 1 + \sum_{i=2}^n T_i$, implying that 
$E[K] = 1 + (n-1)p_n = np_n + (1-p_n)$.



\item Given $K$, let $[\pi_1,...,\pi_K]$ be a uniformly distributed random permutation on $[1,2,\ldots, K]$.
The output segments are then obtained by setting, for $i=1,...,K$,
$
\vY_i = \vX_{\pi_i}.
$

\end{enumerate}

We note that there are no bit-level errors, \emph{e.g.}, bit flips, in this process. 
We also point out that we allow the tearing probability to be a function of the block length $n$, thus, including subscript $n$ in $p_n$.



A code with rate $R$ for the torn-paper channel is a set $\C$ of $2^{nR}$ binary codewords, each of length $n$, together with a decoding procedure that maps a set $\Y$ of variable-length binary strings to an index $\hat W \in \{1,...,2^{nR}\}$.
The message $W$ is assumed to be chosen uniformly at random from $\{1,...,2^{nR}\}$, and the error probability of a code is defined accordingly.
A rate $R$ is said to be achievable if there exists a sequence of rate-$R$ codes $\{\C_n\}$, with blocklength $n \to \infty$, whose error probability tends to $0$ as $n \to \infty$.
The capacity $C$ is defined as the supremum over all achievable rates.
Notice that $C$ should be a function of the sequence of tearing probabilities $\{p_n\}_{n=1}^\infty$.

\vspace{3mm}

\noindent {\bf Notation:}
Throughout the paper, $\log(\cdot)$ represents the logarithm base $2$, while $\ln(\cdot)$ represents the natural logarithm.
For functions $f(n)$ and $g(n)$, we write $g(n) = o(f(n))$ if $g(n)/f(n) \to 0$ as $n \to \infty$.
For an event $A$, we let $\1_A$ or $\1\{A\}$ be the binary indicator of $A$.

\section{Main Results}
\label{Section:Main}



If the encoder had access to the tearing locations ahead of time, a natural coding scheme would involve placing unique indices on every fragment, and using the remaining bits for encoding a message.
In particular, if the message block broke evenly into $n p_n$ pieces of length $[N_1] = 1/p_n$, results from \cite{noisyshuffling} imply that placing a unique index of length $\log(np_n)$ in each fragment is capacity optimal. 
In this case, the capacity is $(1-\alpha)^+$, where $\alpha = \lim_{n\to \infty} p_n\log n$ (assuming the limit exists).
If $\alpha \geq 1$, no positive rate is achievable.


However, in our setting, the fragment lengths are random and the same index-based approach cannot be used.
Because we do not know the tearing points, we cannot place indices at the beginning of each fragment.
Furthermore, while the expected fragment length may be long, some fragments may be shorter than $\log (np_n)$ and a unique index could not be placed in them even if we knew the tearing points.
Our main result shows that, surprisingly, the random tearing locations and fragment lengths in fact increases the channel capacity.




\begin{theorem} \label{thm:capacity}
The capacity of the torn-paper channel is 
\aln{
C = e^{-\alpha},
}
where $\alpha = \lim_{n \to \infty} {p_n \log n}$.
\end{theorem}


In Sections~\ref{Section:Converse} and \ref{Section:Achievability} we prove Theorem~\ref{thm:capacity}.
To prove the converse to this result, we exploit the fact that, for large $n$, $N_i/\log n$ has an approximately exponential distribution.
This, together with several concentration results, allows us to partition the set of fragments into multiple bins of fragments with roughly the same size and view the torn-paper coding, in essence, as parallel channels with fixed-size fragments.
Our achievability is based on random coding arguments and does not provide much insight into efficient coding schemes.
This opens up interesting avenues for future research.


\section{Converse}
\label{Section:Converse}

In order to prove the converse, we 
first partition the input and output strings based on length.
This allows us to view the torn-paper channel as a set of parallel channels, each of which involves fragments of roughly the same size.
More precisely, for an integer parameter $L$, we will let
\al{
& \X_k = \left\{ \vX_i : \tfrac{k-1}{L} \log n \leq N_i < \tfrac{k}{L} \log n \right\}
\text{ and } \nonumber \\
& \Y_k = \left\{ \vY_i : \tfrac{k-1}{L} \log n \leq N_{\pi_i} < \tfrac{k}{L} \log n \right\},
\label{eq:ykdef}
}
for $k=1,2,...$,
and we will think of the transformation from $\X_k$ to $\Y_k$ as a separate channel.
Notice that the $k$th channel is intuitively similar to the shuffling channel with equal-length pieces considered in \cite{DNAStorageISIT}.

We will use the fact that the number of fragments in $\Y_k$ concentrates as $n \to \infty$.
More precisely, we let
\al{
q_{k,n} = \Pr\left( \frac{k-1}{L} \leq \frac{N_1}{\log n} < \frac{k}{L}  \right),
}
and we have the following lemma, proved in Section~\ref{Section:appendix}.
\begin{lemma} \label{lem:Ykconc}
For any $\ep > 0$ and $n$ large enough,
\al{
\Pr\left( \left| |\Y_k| - n p_n q_{k,n} \right| > \ep n p_n \right)  
\leq 4 e^{-n p_n^2 \ep^2 / 4},
}
\end{lemma}

Notice that, since $\lim_{n\to\infty} p_n \log n = \alpha$, $E\left[\frac{N_1}{\log n}\right] \to \alpha^{-1}$ as $n \to \infty$.
Moreover, asymptotically, $\frac{N_1}{\log n}$ approaches an ${\rm Exponential}(\alpha)$ distribution.
This known fact is stated as the following lemma, which we also prove in Section~\ref{Section:appendix}.
\begin{lemma} 
If $N^{(n)}$ is a ${\rm Geometric}(p_n)$ random variable and $\lim_{n\to\infty} E[N^{(n)}]/\log n = 1/\alpha$, then
\al{
\lim_{n\to \infty}
\Pr\left( N^{(n)} \geq \beta \log n \right) = e^{-\alpha \beta}.
}
\label{lem:exp_app1}
\end{lemma}
Lemma~\ref{lem:Ykconc} implies that $E[|\Y_k|] = np_nq_{k,n} + o(np_n)$, and 
\al{
\lim_{n\to \infty} \frac{E\left[|\Y_k|\right]}{np_n} 
& = \lim_{n\to \infty} \frac{ np_nq_{k,n} + o(np_n)}{np_n} \nonumber \\
& = \lim_{n\to\infty} \Pr\left( \tfrac{k-1}{L} \leq \tfrac{N_1}{\log n} < \tfrac{k}{L}  \right) \nonumber \\
& = e^{-\alpha (k-1)/L} - e^{-\alpha k/L},
}
where the last equality follows from Lemma~\ref{lem:exp_app1}.
Next, we define event 
$\E_{k,n} = \{ \left| |\Y_k| - np_nq_{k,n} \right| > \ep_n n p_n \}$, where $\ep_n = 1/\log(n)$, which guarantees that, as $n\to \infty$, $\ep_n \to 0$ and $\Pr(\E_{k,n}) \to 0$ from Lemma~\ref{lem:Ykconc}.
Then,
\al{
 H(\Y_k) & \leq H(\Y_k,\1_{\E_{k,n}}) \leq 1 + H(\Y_k | \1_{\E_{k,n}} ) \nonumber \\
 & \leq 1 + 2n \Pr(\E_{k,n}) + H(\Y_k | \bar \E_{k,n} ),
 \label{eq:bound1}
}
where we loosely upper bound $H(\Y_k | \E_k)$ with $2n$, since $\Y$ can be fully described by the binary string $X^n$ and the $n-1$ tearing points indicators $T_2,...,T_n$. 

In order to bound $H(\Y_k | \bar \E_{k,n} )$, \emph{i.e.}, the entropy of $\Y_k$ given that its size is close to $np_nq_{k,n}$, we first note that the number of possible distinct sequences in $\Y_k$ is 
\aln{
\sum_{i=\frac{k-1}{L}\log n}^{\frac{k}{L}\log n} 2^i < 2 \cdot 2^{\frac{k}{L}\log n} = 2n^{k/L}.
}
Moreover, given $\bar \E_k$, 
\al{
 |\Y_k| & \leq np_nq_{k,n} + \ep n p_n  \nonumber \\
 & = np_n \left[ \ep + \Pr\left( \frac{k-1}{L} \leq \frac{N_1}{\log n} < \frac{k}{L}  \right) \right]
 \triangleq M,
 \label{eq:M}
}
and the set $\Y_k$ can be seen as a histogram $(x_1,...,x_{2n^{k/L}})$ over all possible $2n^{k/L}$  strings with $\sum x_i = M$.
Notice that we can view the last element of the histogram as containing ``excess counts'' if $|\Y_k| < M$.
Hence, from Lemma 1 in \cite{DNAStorageISIT},
\al{
 H(\Y_k | \bar \E_{k,n} ) & \leq \log {{2n^{k/L} + M -1}\choose{M}} \nonumber \\ 
 & \leq M \log \left(\frac{e(2n^{k/L}+M-1)}{M} \right) \nonumber \\ 
 & = M \left[ \log \left(2n^{k/L}+M-1\right) + \log(e) - \log M \right] \nonumber \\ 
& = M \left[ \max(\tfrac{k}{L}\log n,\log M) - \log M + o(\log n) \right] \nonumber \\
& = M \left[ (\tfrac{k}{L}\log n - \log M)^+ + o(\log n) \right] \nonumber \\
& = M \log n \left[ 
 (\tfrac{k}{L}  - \log M/\log n)^+ + o(1)
\right].
\label{eq:bound2}
}
From (\ref{eq:M}), we have
${\log M}/{\log n} \to 1$ as $n\to \infty$.
Combining (\ref{eq:bound1}) and (\ref{eq:bound2}), dividing by $n$, and letting $n \to \infty$ yields
\al{
\lim_{n\to \infty} & \frac{H(\Y_k)}{n}
= \lim_{n \to \infty} \frac{H(\Y_k|\bar \E_{k,n}) + 1 + 2n \Pr(\E_{k,n})}{n}
\nonumber \\
& \leq \lim_{n \to \infty} \frac{M \log n}{n} \left(\frac{k}{L} - 1\right)^+ \nonumber \\
& = \lim_{n \to \infty} p_n \log n \, ( q_{k,n} + \ep_n  ) \left(\frac{k}{L} - 1\right)^+
\nonumber \\
& = \alpha \left( e^{-\alpha (k-1)/L} - e^{-\alpha k/L} \right) \left(\frac{k}{L} - 1\right)^+. 
\label{eq:limyk}
}
In order to bound an achievable rate $R$, we use Fano's inequality to obtain
\al{
n R & \leq I(X^n;\Y) + o(n) \leq H(\Y) + o(n), 
\label{eq:fano}
}
and we conclude that any achievable rate must satisfy
$
R \leq \lim_{n\to \infty} \frac{H(\Y)}{n}.
$
In order to connect (\ref{eq:fano}) and (\ref{eq:limyk}), we state the following lemma,
which allows us to move the limit inside the summation.
The proof is in Section~\ref{Section:appendix}.
\begin{lemma} \label{lem:entropy}
If $\Y_k$ is defined as in (\ref{eq:ykdef}) for $k=1,2,...$,
\aln{
\lim_{n\to \infty} \frac{H(\Y)}{n} \leq 
\sum_{k=1}^\infty \lim_{n\to \infty}  \frac{H(\Y_k)}{n}.
}
\end{lemma}
Using this lemma and (\ref{eq:limyk}), 
we can upper bound any achievable rate as
\al{
R & \leq  \lim_{n\to\infty}   \frac{H(\Y)}{n} 
\leq \sum_{k=1}^\infty \lim_{n\to \infty}  \frac{H(\Y_k)}{n}
\nonumber \\
&
= \sum_{k=L+1}^\infty \alpha \left( e^{-\alpha (k-1)/L} - e^{-\alpha k/L} \right) (\tfrac{k}{L} - 1) \nonumber \\
& = \frac{\alpha}{L} \sum_{k=L+1}^\infty  k \left( e^{-\alpha (k-1)/L} - e^{-\alpha k/L} \right) \nonumber \\
& \quad - \alpha \sum_{k=L+1}^\infty  \left( e^{-\alpha (k-1)/L} - e^{-\alpha k/L} \right) \nonumber \\
& = \frac{\alpha}{L} \sum_{k=L+1}^\infty  k \left( e^{-\alpha (k-1)/L} - e^{-\alpha k/L} \right) - \alpha e^{-\alpha},
}
where the last equality is due to a telescoping sum.
The remaining summation can be computed as
\al{
\sum_{k=L+1}^\infty &  k \left( e^{-\alpha (k-1)/L} - e^{-\alpha k/L} \right) \nonumber \\
& = (L+1)e^{-\alpha} + \sum_{k=L+2}^\infty e^{-\alpha (k-1)/L} \nonumber \\
& = L e^{-\alpha} + e^{-\alpha} \sum_{k=0}^\infty e^{-\alpha k/L} 
= L e^{-\alpha} + \frac{e^{-\alpha}}{1-e^{-\alpha/L}}. 
\nonumber
}
We conclude that any achievable rate must satisfy
\al{
R & < \frac{\alpha}{L}  \left(Le^{-\alpha} + \frac{e^{-\alpha}}{1-e^{-\alpha/L}} \right) - \alpha e^{-\alpha} 
= \frac{\alpha e^{-\alpha }}{L(1-e^{-\alpha/L})},
\nonumber 
}
for any positive integer $L$.
Since 
\aln{
\lim_{L \to \infty} L(1-e^{-\alpha/L}) = \alpha,
}
we obtain the outer bound $R < e^{-\alpha}$.

\section{Achievability via Random Coding}
\label{Section:Achievability}



A random coding argument can be used to show that any rate $R < e^{-\alpha}$
is achievable.
Consider generating a codebook $\C$ with $2^{nR}$ codewords, by independently picking each symbol as ${\rm Bernoulli}(1/2)$.
Let $\C = \{\x_1,...,\x_{2^{nR}}\}$, where $\x_i$ is the random codeword associated with message $W = i$.
Notice that optimal decoding can be obtained by simply finding an index $i$ such that $\x_i$ corresponds to a concatenation of the strings in $\Y$.
If more than one such codewords exist, an error is declared.

Suppose message $W=1$ is chosen and $\Y = \{\vY_1,...,\vY_K\}$ is the random set of output strings.
To bound the error probability we consider a suboptimal decoder that throws out all fragments shorter than $\gamma \log n$, for some $\gamma > 0$ to be determined, and simply tries to find a codeword $\x_i$ that contains all output strings $\Y_\gamma = \{\vY_i : N_{\pi_i} \geq \gamma \log n\}$ as non-overlapping substrings.
If we let $\E$ be the error event averaged over all codebook choices, we have
\aln{
& \Pr(\E) = \Pr(\E | W = 1) \\
& = \Pr\left( \text{some $\x_j,~j\neq1,$ contains all strings in } \Y_\gamma | W = 1\right).
}

Using a similar approach to the one used in Section~\ref{Section:Converse}, it can be shown that
 $E[|\Y_\gamma|] = np_n \Pr(N_1 \geq \gamma \log n ) + o(np_n)$.
From Lemma~\ref{lem:exp_app1}, we thus have
\al{
\lim_{n \to \infty} \frac{E[|\Y_\gamma|]}{n\cdot p_n} =
\lim_{n\to\infty} \Pr\left( N_1 \geq \gamma \log n \right) = e^{-\alpha \gamma}.
\label{eq:gamma}
}
If we let $Z_i$ be the binary indicator of the event $\{  N_i \geq \gamma \log n \}$, then
$|\Y_\gamma| = \sum_{i=1}^K Z_i$.
In Section~\ref{Section:appendix}, 
we prove the following concentration result.
\begin{lemma}\label{lem:gammaconc}
For any $\ep > 0$, as $n \to \infty$,
\al{
\Pr \left( ||\Y_\gamma| - e^{-\alpha \gamma}np_n| > \ep n p_n \right)  \to 0.
\label{eq:gammaconc}
}
\end{lemma}

In addition to characterizing $|\Y_\gamma|$ asymptotically, we will also be interested in the total length of the sequences in $\Y_\gamma$.
Intuitively, this determines how well the fragments in $\Y_{\gamma}$ \emph{cover} their codeword of origin $\x_1$.
\begin{definition} \label{def:cov}
The coverage of $\Y_\gamma$ is defined as 
\al{
c_\gamma = \frac{1}{n}\sum_{i=1}^K N_i \1_{\{N_i \geq \gamma \log n\}.
}
\label{eq:covdef}}
Notice that $0 \leq c_\gamma \leq 1$ with probability $1$.
\end{definition}
In order to characterize $c_\gamma$ asymptotically, we will again resort to the  exponential approximation to a geometric distribution, through the following lemma.
\begin{lemma} \label{lem:exp2}
If $N^{(n)}$ is a ${\rm Geometric}(p_n)$ random variable and $\lim_{n\to\infty} E[N^{(n)}]/\log n = 1/\alpha$, then, for any $\beta \geq 0$,
\al{
\lim_{n \to \infty} & E\left[ N^{(n)} \1_{\{ N^{(n)} \geq \gamma \log n  \}}\right] / \log n \nonumber \\
& = E\left[ \tilde N \1_{\{ \tilde N \geq \gamma   \}}\right]
= \left(\gamma + \frac1\alpha\right)e^{-\alpha \gamma},
}
where $\tilde N$ is an ${\rm Exponential}(\alpha)$ random variable.
\end{lemma}
Using Lemma~\ref{lem:exp2}, we can characterize the asymptotic value of $E[c_\gamma]$ and show that $c_\gamma$ concentrates around this value.
More precisely, 
we show the following lemma in Section~\ref{Section:appendix}.
\begin{lemma} \label{lem:covconc}
For any $\ep > 0$, as $n\to \infty$,
\al{  \label{eq:covconc}
\Pr & \left( \left| c_\gamma - (\alpha \gamma + 1) e^{-\alpha \gamma} \right| > \ep  \right)
\to 0.
}
\end{lemma}

In particular, Lemma~\ref{lem:covconc} implies that
\al{
\lim_{n\to\infty} E[c_\gamma]
= (\alpha \gamma + 1) e^{-\alpha \gamma},
\label{eq:coverage}
}
and that $c_\gamma$ cannot deviate much from this value with high probability.
%
%
%
%
%
If we let $B_1 = (1+\ep) e^{-\alpha \gamma} n p_n$ and 
$B_2 = (1-\ep) (\alpha \gamma + 1)e^{-\alpha \gamma}$, and we define
the event 
\al{
\B & = \{ |\Y_\gamma|  >  B_1 \} 
\cup 
\{ c_\gamma < B_2 \},
}
then (\ref{eq:gammaconc}) and 
(\ref{eq:covconc}) imply that 
$\Pr(\B) \to 0$ as $n \to \infty$.
Since $\B$ is independent of $\{W=1\}$,
we can upper bound the probability of error as
\aln{
\Pr(\E) & \leq \Pr\left( \text{some $\x_j$ contains all strings in } \Y_\gamma | W = 1\right) \\
& \leq \Pr\left( \text{some $\x_j$ contains all strings in } \Y_\gamma | \bar \B, W = 1\right) \nonumber \\ 
& \qquad + \Pr(\B) \\
& \leqnum |\C| \frac{n^{B_1}}{2^{nB_2}} + \Pr(\B) \nonumber \\
& \leq 2^{nR} \, {2^{B_1 \log n}} \, 2^{-nB_2} + o(1) \nonumber \\
& = 2^{nR} \, 2^{(1+\ep) e^{-\alpha \gamma} n p_n \log n - n(1-\ep)(\alpha \gamma + 1) e^{-\alpha \gamma}} + o(1) \nonumber \\
& = 2^{-n((1-\ep) (\alpha \gamma + 1) e^{-\alpha \gamma} -(1+\ep) e^{-\alpha \gamma} p_n \log n -R)} + o(1).
\rescnt
}
Inequality \cnt follows from the union bound and from the fact that thre are at most $n^{B_1}$ ways to align the strings in $\Y_\gamma$ to a codeword $\x_j$ in a non-overlapping way and, given this alignment, $2^{nB_2}$ bits in $\x_j$ must be specified.
Since $p_n \log n \to \alpha$ as $n\to \infty$,
we see that we can a rate $R$ as long as
\aln{
R < (1-\ep)(1+ \alpha \gamma) e^{-\alpha \gamma} - (1+\ep)\alpha e^{-\alpha \gamma},
}
for some $\ep > 0$ and $\gamma> 0$.
Letting $\ep \to 0$,
yields 
\aln{
R < (1+ \alpha \gamma - \alpha)e^{-\alpha \gamma}
}
for some $\gamma > 0$.
The right-hand side is maximized by setting $\gamma = 1$, which implies that we can achieve any rate $R < e^{-\alpha}$.



\section{Proofs of Lemmas}
\label{Section:appendix}

\begin{lemmarep}
The number of fragments in $\Y_k$ satisfies
\aln{
\Pr\left( \left| |\Y_k| - np_nq_{k,n} \right| > \ep n p_n \right) \leq 
4 e^{-n p_n^2 \ep^2 / 4},
}
for any $\ep > 0$ and $n$ large enough.
\end{lemmarep}

\begin{proof}[Proof of Lemma~\ref{lem:Ykconc}]
First notice that, 
since $K = 1 + \sum_{i=2}^n T_i$, where $T_2,...,T_n$ are i.i.d.~${\rm Bernoulli}(p_n)$ random variables, 
$
E[K] = np_n + (1-p_n),
$ 
and using Hoeffding's inequality,
\al{
\Pr( & | K - np_n |  > \delta np_n ) 
\nonumber \\ & 
= \Pr\left( \left| K-E[K] + (1-p_n) \right|  > \delta np_n \right) 
\nonumber \\
& \leq \Pr\left( \left| K-E[K] \right|  > \delta np_n - (1-p_n) \right) 
\nonumber \\
& = \Pr\left( \left| \sum_{i=2}^n (T_i-p_n) \right|  > (n-1) \frac{\delta np_n - (1-p_n)}{n-1} \right) 
\nonumber \\
& \leq 2 e^{-2 (n-1) \left(\frac{\delta n p_n - (1-p_n)}{n-1} \right)^2} 
\leq 2 e^{-2 n \left(\frac{\delta n p_n - (1-p_n)}{n} \right)^2} 
\nonumber \\
& 
\leq 2 e^{-n p_n^2 \delta^2},
\label{eq:Kconc}
}
where the last inequality holds for $n$ large enough.


Now suppose the sequence $N_1,N_2,...$ of independent ${\rm Geometric}(p_n)$ random variables is an infinite sequence (and does not stop at $K$).
Let $Z_i$ be the binary indicator of the event $\{ (k-1)/L \leq N_i/\log n < k/L \}$, and $\tilde Z = \sum_{i=1}^{np_n} Z_i$.
Intuitively, $|\Y_k|$ and $\tilde Z$ should be close.
In particular, 
$
||\Y_k|-\tilde Z| \leq |K - np_n|.
$
Moreover, $E[\tilde Z] = np_n q_{k,n}$.
If $|\tilde Z - np_n q_{k,n}| < \tfrac12 \ep n p_n $ and $||\Y_k|-\tilde Z| < |K - np_n| < \tfrac12 \ep n p_n$, by the triangle inequality, $||\Y_k| - np_n q_{k,n}| < \ep n p_n$.
Therefore,
\aln{
\Pr & \left( \left| |\Y_k| - np_n q_{k,n} \right| > \ep n p_n \right) \\
& \leq \Pr\left( |\tilde Z - np_n q_{k,n}| > \tfrac12 \ep n p_n \right) \\ 
& \quad \quad \quad + \Pr\left( \left| K - np_n \right| >  \tfrac12 \ep n p_n \right) \\
& \leq 2 e^{-n p_n \ep^2 / 2} + 2 e^{-n p_n^2 \ep^2 / 4} 
\leq 4 e^{-n p_n^2 \ep^2 / 4}
}
where we used Hoeffding's inequality and (\ref{eq:Kconc}).
\end{proof}

\begin{lemmarep} \label{lem:exp1}
If $N^{(n)}$ is a ${\rm Geometric}(p_n)$ random variable and $\lim_{n\to\infty} E[N^{(n)}]/\log n = 1/\alpha$, then
\aln{
\lim_{n\to \infty}
\Pr\left( N^{(n)} \geq \beta \log n \right) = e^{-\alpha \beta}.
}
\end{lemmarep}

\begin{proof}[Proof of Lemma~\ref{lem:exp1}]
By definition, 
\aln{
& \Pr\left( N^{(n)} \geq \beta \log n \right) 
 = (1-p_n)^{\beta \log n} \nonumber \\
& \quad \quad = \left(1-\frac{1}{E[N^{(n)}]}\right)^{E[N^{(n)}] (\beta \log n/E[N^{(n)}])}.
}
As $n \to \infty$, $\log n / E[N^{(n)}] \to \alpha$ and $E[N^{(n)}] \to \infty$. 
Hence, $(1-1/E[N^{(n)}])^{E[N^{(n)}]} \to e^{-1}$,
implying the lemma.
\end{proof}

\begin{lemmarep}
If $\Y_k$ is defined as in (\ref{eq:ykdef}) for $k=1,...,\infty$,
\aln{
\lim_{n\to \infty} \frac{H(\Y)}{n} \leq 
\sum_{k=1}^\infty \lim_{n\to \infty}  \frac{H(\Y_k)}{n}.
}
\end{lemmarep}

\begin{proof}[Proof of Lemma~\ref{lem:entropy}]
For a fixed integer $A$, we define $\Y_{\geq A} = \{ \vY_i : N_{\pi_i} \geq (A/L) \log n \}$ and we have
\al{
\lim_{n\to \infty} \frac{H(\Y)}{n} & \leq
\lim_{n\to \infty} \sum_{k=1}^A \frac{H(\Y_k)}{n} + \lim_{n\to \infty} \frac{H(\Y_{\geq A})}{n} \nonumber \\
& = \sum_{k=1}^A \lim_{n\to \infty} \frac{H(\Y_k)}{n} + \lim_{n\to \infty} \frac{H(\Y_{\geq A})}{n}.
\label{eq:upperA}
}
If we define $c_\gamma$ as in Definition~\ref{def:cov},
from Lemma~\ref{lem:covconc}, we have
\aln{
\lim_{n\to\infty} E\left[ c_{A/L} \right]
=  (\alpha A/L + 1) e^{-\alpha A/L}.
}
Moreover, from Lemma~\ref{lem:covconc},
the event 
\aln{
\A = \{ c_{A/L} > (\alpha A/L + 1) e^{-\alpha A/L} + \delta \}
}
has vanishing probability as $n \to \infty$.
This allows us to write
\aln{
H(\Y_{\geq A}) & \leq H(\Y_{\geq A}|\bar \A) +  H(\Y_{\geq A}| \A) \Pr(\A) + 1 \\
& \leq H(\Y_{\geq A}|\bar \A) +  2 n \Pr(\A) + 1 \\
& \leq 2n\left[(\alpha A/L + 1) e^{-\alpha A/L} + \delta \right] + o(n).
}
Hence, from (\ref{eq:upperA}), we have that for every $A$ and $\delta > 0$,
\aln{
\lim_{n\to \infty} \frac{H(\Y)}{n} \leq \sum_{k=1}^A \lim_{n\to \infty} \frac{H(\Y_k)}{n} + 2 (\alpha A/L + 1) e^{-\alpha A/L} + 2 \delta.
}
Notice that $(\alpha A/L + 1) e^{-\alpha A/L}  \to 0$ as $A \to \infty$.
Therefore, we can let $\delta \to 0$ and $A \to \infty$, and we conclude that
\aln{
\lim_{n\to \infty} \frac{H(\Y)}{n} \leq 
\sum_{k=1}^\infty \lim_{n\to \infty}  \frac{H(\Y_k)}{n}.
}
\end{proof}

\begin{lemmarep}
The number of fragments in $\Y_\gamma$ satisfies
\aln{
\Pr \left( ||\Y_\gamma| - e^{-\alpha \gamma}np_n| > \ep n p_n \right) \leq
4 e^{-n p_n^2 \ep^2 / 9}
}
for any $\ep >0$ and $n$ large enough.
\end{lemmarep}

\begin{proof}[Proof of Lemma~\ref{lem:gammaconc}]
Let $Z_i = \1_{\{  N_i \geq \gamma \log n \}}$, for $i=1,2,...$.
Then
$|\Y_\gamma| = \sum_{i=1}^K Z_i$. 
Since $K$ is random (and not independent of the $N_i$s), we need to follow similar  steps to those in the proof of Lemma~\ref{lem:Ykconc}.

Let us assume that the sequence $N_1,N_2,...$ of independent ${\rm Geometric}(p_n)$ random variables is an infinite sequence and let
$\tilde Z = \sum_{i=1}^{np_n} Z_i$.
Notice that 
$\tilde Z$ is a sum of i.i.d.~Bernoulli random variables with
\al{
E[\tilde Z] = np_n \Pr(N_1 \geq \gamma \log n ),
}
and the standard Hoeffding's inequality can be applied.
Moreover, from Lemma~\ref{lem:exp1}, 
\aln{
\lim_{n\to\infty} E[\tilde Z]/(np_n) = e^{-\alpha \gamma}} 
and, for any $\delta > 0$,
$
| E[\tilde Z] - e^{-\alpha \gamma} np_n |
< \delta n p_n,
$
for $n$ large enough.
If we set $\delta = \ep/3$ and, for $n$ large enough, we have
$
| E[\tilde Z] - e^{-\alpha \gamma} np_n |
< \tfrac13 \ep n p_n,
$.
Moreover, if $|\tilde Z - E[\tilde Z]| < \tfrac13 \ep n p_n $ and $||\Y_\gamma|-\tilde Z| < |K - np_n| < \tfrac13 \ep n p_n$,  by the triangle inequality (applied twice), $||\Y_\gamma| - e^{-\alpha \gamma}np_n| < \ep n p_n$.
Hence,
\aln{
& \Pr \left( ||\Y_\gamma| - e^{-\alpha \gamma}np_n| > \ep n p_n \right) \\
& \leq \Pr\left( |\tilde Z - E[\tilde Z]| > \tfrac13 \ep n p_n \right) + 
\Pr\left( \left| |\Y_\gamma| - \tilde Z  \right| >  \tfrac13 \ep n p_n \right) \\
& \leq \Pr\left( \left| \tilde Z - E|Z| \right| > \tfrac13 \ep n p_n \right) + 
\Pr\left( \left| K - np_n \right| >  \tfrac13 \ep n p_n \right) \\
& \leq 2 e^{-2 n p_n \ep^2 / 9} 
+ 2 e^{-n p_n^2 \ep^2 / 9} 
\leq 4 e^{-n p_n^2 \ep^2 / 9}
}
where we used Hoeffding's inequality and  (\ref{eq:Kconc}).
\end{proof}

\begin{lemmarep}
If $N^{(n)}$ is a ${\rm Geometric}(p_n)$ random variable and $\lim_{n\to\infty} E[N^{(n)}]/\log n = 1/\alpha$, then, for any $\beta \geq 0$,
\aln{
\lim_{n \to \infty} & E\left[ N^{(n)} \1_{\{ N^{(n)} \geq \gamma \log n  \}}\right] / \log n \nonumber \\
& = E\left[ \tilde N \1_{\{ \tilde N \geq \gamma   \}}\right]
= \left(\gamma + \frac1\alpha\right)e^{-\alpha \gamma},
}
where $\tilde N$ is an ${\rm Exponential}(\alpha)$ random variable.
\end{lemmarep}

\begin{proof}[Proof of Lemma~\ref{lem:exp2}]
We first notice that
\aln{
\frac{1}{\log n} & E \left[ {N^{(n)}} \1_{\{ N^{(n)} \geq \gamma \log n  \}}\right]  \\ 
& = \frac{1}{\log n} E \left[ {N^{(n)}} \left|  N^{(n)} \geq \gamma \log n \right. \right] 
\Pr\left(  N^{(n)} \geq \gamma \log n  \right) \\
& = \frac{1}{\log n} \left( \lceil \gamma \log n \rceil + E[N^{(n)}] \right)
\Pr\left(  N^{(n)} \geq \gamma \log n  \right),
}
where we used the memoryless property of the Geometric distribution.
As $n \to \infty$, we have $\lceil \gamma \log n\rceil / \log n \to \gamma$, $E[N^{(n)}]/\log n \to 1/\alpha$.
Moreover, from Lemma~\ref{lem:exp_app1}, $\Pr\left(  N^{(n)} \geq \gamma \log n  \right) \to e^{-\alpha \gamma}$, and the lemma follows.
\end{proof}

\begin{lemmarep}
If $c_\gamma$ is defined as in (\ref{eq:covdef}), then, for any $\ep > 0$,
\aln{ 
\Pr & \left( \left| c_\gamma - (\alpha \gamma + 1) e^{-\alpha \gamma} \right| > \ep  \right)
\leq \frac{19}{\ep^2 n p_n^2}
}
for $n$ large enough.
\end{lemmarep}

\begin{proof}[Proof of Lemma~\ref{lem:covconc}]
Since $c_\gamma = \frac{1}{n}\sum_{i=1}^K N_i \1_{\{N_i \geq \gamma \log n\}}$, where $K$ is a random variable, we once again follow 
an approach similar to the one in the proof of Lemma~\ref{lem:Ykconc}.

Let us assume that the sequence $N_1,N_2,...$ of independent ${\rm Geometric}(p_n)$ random variables is an infinite sequence.
Let $Z_i= N_i \1_{\{N_i \geq \gamma \log n\}}$, 
and $\tilde Z = \sum_{i=1}^{np_n} Z_i$.
Since $E[\tilde Z] = np_n E[ N \1_{\{ N_1 \geq \gamma \log n\}}]$, by Lemma~\ref{lem:exp2}, 
\al{
\lim_{n\to \infty}\frac{E[\tilde Z]}{n} \to \alpha \left( \gamma + \frac1{\alpha} \right) e^{-\alpha \gamma}.
\label{eq:ztilde}
}

Intuitively, $Z = n c_\gamma$ and $\tilde Z$ should be close.
If $\tilde Z > Z$, then $np_n > K$, and
\al{
|Z-\tilde Z| = \sum_{i=K+1}^{np_n} Z_i \leq \sum_{i=K+1}^{np_n} N_i \leq \left|\sum_{i=1}^{np_n}N_i - n \right|.
}
If $Z > \tilde  Z$, then $K> np_n$, and
\al{
|Z-\tilde Z| = \sum_{i=np_n+1}^{K} Z_i \leq \sum_{i=np_n+1}^{K} N_i \leq \left|\sum_{i=1}^{np_n}N_i - n \right|.
}
Hence, for any $\delta > 0$, we have that
\al{
\Pr\left( |Z-\tilde Z| > \delta n p_n \right) & \leq \Pr\left( \left|\sum_{i=1}^{np_n}N_i - n \right| > \delta n p_n \right) \nonumber \\
& \leq e^{-n p_n (\delta - \ln (1+\delta))}
+ e^{-n (-\delta - \ln (1-\delta))} \nonumber \\
& \leq 2 e^{-n p_n (\delta - \ln (1+\delta))}.
\label{eq:app2}
}
where we used the Chernoff bound for exponentially distributed random variables \cite{janson2018tail}, and the fact that
$x - \ln(1+x) < -x - \ln(1-x)$ for $x>0$.

To bound the probability that $|\tilde Z - E[\tilde Z]| > \delta n$, we can use a Chernoff bound, which requires the computation of the rate function for $N_1 \1_{\{N_1 \geq \gamma \log n\}}$.
A simpler approach is to use Chebyshev's inequality, which yields
\al{
\Pr & \left( |\tilde Z-E[\tilde Z]| > \delta n \right)
 \leq \frac{{\rm Var}(Z_1)}{\delta^2 n} 
\leq \frac{E[Z_1^2]}{\delta^2 n} 
\nonumber \\
& = \frac{E[N_1^2 \1_{\{N_1 \geq \gamma \log n\}}]}{\delta^2 n} 
\leq \frac{E[N_1^2]}{\delta^2 n}
= \frac{2-p_n}{\delta^2 n p_n^2}.
\label{eq:chebyshev}
}

From (\ref{eq:ztilde}), we know that for any $\delta>0$ and $n$ large enough,
\aln{
|E[\tilde Z] - n(\alpha \gamma + 1) e^{-\alpha \gamma}| < \delta n.
\label{eq:coverage}
}
Moreover, if $|\tilde Z - E[\tilde Z]| < \tfrac13 \ep n$, $|nc_\gamma-\tilde Z| < \tfrac13 \ep n$, and $|E[\tilde Z] - n(\alpha \gamma + 1) e^{-\alpha \gamma}| <  \tfrac13 \ep n$, then, by the triangle inequality, $|c_\gamma - (\alpha \gamma + 1)e^{-\alpha \gamma}| < \ep$.
Therefore, for $n$ large enough so that $|E[\tilde Z] - n(\alpha \gamma + 1) e^{-\alpha \gamma}| <  \tfrac13 \ep n$,
\aln{
\Pr & \left( \left| c_\gamma - (\alpha \gamma + 1) e^{-\alpha \gamma} \right| > \ep  \right) \\
& \leq \Pr\left( | \tilde Z - E[\tilde Z] | > \tfrac13 \ep n \right) + 
\Pr\left( | \tilde Z - Z | >  \tfrac13 \ep n \right) \\
& \leq \Pr\left( | \tilde Z - E[\tilde Z] | > \tfrac13 \ep n \right) + 
\Pr\left( | \tilde Z - Z | >  \tfrac13 \ep n p_n \right) \\
& \leq {18}/({\ep^2 n p_n^2})+ 2 e^{-n p_n (\ep/3 - \ln (1+\ep/3))} 
 \leq {19}/({\ep^2 n p_n^2}),
}
where we used (\ref{eq:app2}) and (\ref{eq:chebyshev}), 
and the last inequality follows for $n$ large enough.
\end{proof}

{\footnotesize
\bibliographystyle{ieeetr}
\bibliography{refs.bib}
}

\end{document}